\newif\ifsariel
\sarielfalse 

\documentclass[11pt]{article}

\usepackage[cm]{fullpage}
\usepackage{xspace}
\usepackage{picins}
\usepackage{amsmath,amssymb}
\usepackage{graphicx}
\usepackage{xcolor}
\usepackage{paralist}  

\ifsariel
   \usepackage{hyperref}%
   \hypersetup{%
      breaklinks,%
      ocgcolorlinks,
      colorlinks=true,%
      linkcolor=[rgb]{0.45,0.0,0.0},%
      citecolor=[rgb]{0,0,0.45}
   }
   \usepackage[hypcap=true]{caption}

   \numberwithin{figure}{section}
   \numberwithin{table}{section}
   \numberwithin{equation}{section}

   \usepackage[hyperref,thmmarks]{ntheorem}
   \qedsymbol{\rule{2mm}{2mm}}
   \theoremseparator{.}

   \newcommand{\aftermathA}{\par\vspace{-\baselineskip}}

    \newcommand{\qedhere}{\ifmmode\qed\else\hfill\proofSymbol\fi}

   \newtheorem{theorem}{Theorem}[section]
   \newtheorem{lemma}[theorem]{Lemma}%
   {\theorembodyfont{\rm} \newtheorem{observation}[theorem]{Observation}}

   \newenvironment{proof}[1][Proof]{\trivlist\item[]\emph{#1}:{}\xspace}%
   {\unskip\nobreak\hskip 1em plus 1fil\nobreak%
      \myqedsymbol
      \parfillskip=0pt%
      \endtrivlist}

   \newcommand{\myqedsymbol}{\rule{2mm}{2mm}}
   \qedsymbol{\rule{2mm}{2mm}}

   \numberwithin{figure}{section}
   \numberwithin{table}{section}
   \numberwithin{equation}{section}

   \usepackage{titlesec}%
   \titlelabel{\thetitle. }%

   \definecolor{blue25}{rgb}{0,0,0.55}%

   \newcommand{\emphi}[1]{\emphic{#1}{#1}}   

   \newcommand{\etal}{\textit{et~al.}\xspace}

\else

  \usepackage{amsthm} 
  \usepackage{hyperref}
   \hypersetup{%
      breaklinks,%
      ocgcolorlinks,
      colorlinks=true,%
      linkcolor=[rgb]{0.45,0.0,0.0},%
      citecolor=[rgb]{0,0,0.45}
   }

   \newtheorem{theorem}{Theorem}
   \newtheorem{lemma}[theorem]{Lemma}
   \newtheorem{observation}{Observation}

  \makeatletter
  \partopsep\z@ \textfloatsep 10pt plus 1pt minus 4pt
  \def\section{\@startsection {section}{1}{\z@}%
    {-3.5ex plus -1ex
      minus -.2ex}{2.3ex plus .2ex}{\large\bf}}
  \def\subsection{\@startsection{subsection}{2}%
    {\z@}{-3.25ex plus
      -1ex minus -.2ex}{1.5ex plus .2ex}{\normalsize\bf}}
  \def\@fnsymbol#1{\ensuremath{\ifcase#1\or *\or 1\or 2\or 3\or 4\or
      5\or 6\or 7 \or 8\ or 9 \or 10\or 11 \else\@ctrerr\fi}}
  \makeatother

  \newcommand{\emphi}[1]{\emph{#1}}

  \newcommand{\etal}{{et~al.}\xspace}

  \let\aftermathA\relax
\fi

\let\geq\geqslant
\let\leq\leqslant

\newcommand{\atgen}{\symbol{'100}}

\newcommand{\grapro}{\mathcal{P}}
\newcommand{\edgepro}{\mathcal{P}^{*}}

\newcommand{\Toth}{T\'{o}th\xspace}

\newcommand{\si}[1]{#1}

\newcommand{\face}{\psi}%
\newcommand{\faceA}{\varphi}%

\newcommand{\edge}{{{e}}}

\newcommand{\vertex}{{v}}
\newcommand{\vertexA}{{u}}

\newcommand{\Graph}{{G}}
\newcommand{\GraphA}{{H}}

\newcommand{\Vertices}{{V}}

\newcommand{\Faces}{{\mathcal{F}}}
\newcommand{\Edges}{{E}}

\newcommand{\MeetB}{L}
\renewcommand{\th}{th\xspace}

\begin{document}

\title{On the Number of Edges of Fan-Crossing Free Graphs%
  \thanks{OC and HSK were supported in part by NRF
    grant~2011-0016434 and in part by NRF grant 2011-0030044
    (SRC-GAIA), both funded by the government of Korea.
    SHP was partially supported by NSF AF awards
    CCF-0915984 and CCF-1217462.}}

\author{%
  Otfried Cheong%
  \thanks{Department of Computer Science, KAIST, Daejeon, Korea.
    \si{otfried}\atgen{}kaist.edu.}%
  \and
  Sariel Har-Peled%
  \thanks{%
    Department of Computer Science, University of Illinois, Urbana,
    USA. \si{sariel}\atgen{}\si{illinois.edu}.}
  \and
  Heuna Kim%
  \thanks{Freie \si{Universit\"at} Berlin, Berlin, Germany.
    heunak\atgen{}mi.fu-berlin.de.}%
  \and
  Hyo-Sil Kim%
  \thanks{Department of Computer Science and Engineering, POSTECH,
    Pohang, Korea. hyosil.kim@\atgen{}gmail.com.}}

\maketitle

\begin{abstract}
  A graph drawn in the plane with $n$ vertices is
  \emphi{$k$-fan-crossing free} for $k \geq 2$ if there are no $k+1$
  edges $g,e_1,\dots, e_k$, such that $e_1,e_2,\dots,e_k$ have a
  common endpoint and $g$ crosses all~$e_i$.  We prove a tight bound
  of $4n-8$ on the maximum number of edges of a $2$-fan-crossing free
  graph, and a tight $4n-9$ bound for a straight-edge drawing.  For $k
  \geq 3$, we prove an upper bound of $3(k-1)(n-2)$ edges. We also
  discuss generalizations to monotone graph properties.
\end{abstract}

\section{Introduction}

A \emphi{topological graph}~$\Graph$ is a graph drawn in the plane:
vertices are points, and the edges of the graph are drawn as Jordan
curves connecting the vertices.  Edges are not allowed to pass through
vertices other than their endpoints.  We will assume the topological
graph to be \emphi{simple}, that is, any pair of its edges have at
most one point in common (so edges with a common endpoint do not
cross, and edges cross at most once). Figure~\ref{fig:bad}~(a--b) shows
configurations that are not allowed.

If there are no crossings between edges, then the graph is planar, and
Euler's formula implies that it has at most $3n-6$ edges, where $n$ is
the number of vertices.  What can be said if we relax this
restriction---that is, we permit some edge crossings?

For instance, a topological graph is called \emphi{$k$-planar} if each
edge is crossed at most $k$~times.  Pach and \Toth~\cite{pt-gdfce-97}
proved that a $k$-planar graph on~$n$ vertices has at most
$(k+3)(n-2)$ edges for $0 \leq k \leq 4$, and at most~$4.108
\sqrt{k}n$ edges for general~$k$.  The special case of $1$-planar
graphs has recently received some attention, especially in the graph
drawing community.  Pach and \Toth's bound is~$4n-8$, and this is
tight: starting with a planar graph~$H$ where every face is a
quadrilateral, and adding both diagonals results in a $1$-planar graph
with~$4n-8$ edges.  However, Didimo~\cite{d-dslop-13} showed that
\emph{straight-line} $1$-planar graphs have at most~$4n-9$ edges,
showing that F\'{a}{r{}y}'s theorem does not generalize to $1$-planar
graphs.  This bound is tight, as he constructed an infinite family of
straight-line $1$-planar graphs with $4n-9$~edges.  Hong
\etal~\cite{help-ftopg-12} characterize the $1$-planar graphs that can
be drawn as straight-line $1$-planar graphs.  Grigoriev and
Bodlaender~\cite{gb-agefc-07} showed that testing if a given graph is
$1$-planar is NP-hard.

A topological graph is called \emphi{$k$-quasi planar} if it does not
contain $k$~pairwise crossing edges.  It is conjectured that for any
fixed~$k$ the number of edges of a $k$-quasi planar graph is linear in
the number of vertices~$n$. Agarwal \etal~\cite{aapps-qpglne-1997}
proved this for straight-line $3$-planar graphs, Pach
\etal~\cite{prt-rptg-02} for general $3$-planar graphs,
Ackerman~\cite{a-mnetg-09} for~$4$-planar graphs, and Fox
\etal~\cite{fp-nekqp-13} prove a bound of the form~$O(n
\log^{1+o(1)}n)$ for $k$-planar graphs.

A different restriction on crossings arises in graph drawing: Humans
have difficulty reading graph drawings where edges cross at acute
angles, but graph drawings where edges cross at right angles are
nearly as readable as planar ones. A \emph{right-angle crossing graph}
(RAC graph) is a topological graph with straight edges where edges
that cross must do so at right angle.  Didimo
\etal~\cite{del-dgrac-11} showed that an RAC graph on $n$~vertices has
at most~$4n-10$ edges.  Testing whether a given graph is an RAC graph
is NP-hard~\cite{abs-slrdp-12}.  Eades and Liotta~\cite{el-racgo-13}
showed that an extremal RAC graph, that is, an RAC graph with $n$
vertices and $4n - 10$ edges, is $1$-planar, and is the union of two
maximal planar graphs sharing the same vertex set.

A \emphi{radial $(p,q)$-grid} in a graph~$\Graph$ is a set of $p+q$
edges such that the first $p$~edges are all incident to a common
vertex, and each of the first $p$~edges crosses each of the remaining
$q$~edges.  Pach \etal \cite{ppst-tgnlg-05} proved that a graph
without a radial $(p,q)$-grid, for $p,q \geq 1$, has at most $8 \cdot
24^{q} pn$ edges.

We will call a radial $(k,1)$-grid a \emphi{$k$-fan crossing}.  In
other words, a fan crossing is formed by an edge~$g$ crossing
$k$~edges~$e_1,\dots,e_k$ that are all incident to a common vertex,
see Figure~\ref{fig:bad}~(c). A topological graph is \emphi{$k$-fan-crossing
   free} if it does not contain a $k$-fan crossing. We are
particularly interested in the special case $k=2$. For brevity, let us
call a $2$-fan crossing simply a \emphi{fan crossing} (shown in
Figure~\ref{fig:bad}~(d)), and a $2$-fan-crossing free graph a
\emphi{fan-crossing free graph}.
\begin{figure}[t]
    \centerline{%
       \hfill%
       \includegraphics[page=1]{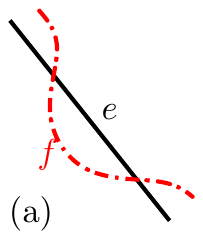}%
       \hfill
       \includegraphics[page=2]{figs/not_allow}
       \hfill
       \includegraphics[page=2]{figs/fan_crossing}
       \hfill
       \includegraphics[page=1]{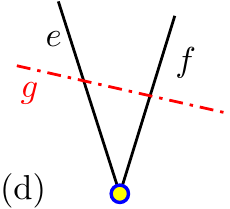}%
       \hfill}
    \caption{(a) and (b) shows illegal embeddings of edges of a
       graph. (c) is a $4$-fan crossing, (d) a fan crossing.}
    \label{fig:bad}
\end{figure}

By Pach \etal{}'s result, a $k$-fan-crossing free graph on
$n$~vertices has at most $192kn$~edges, and a fan-crossing free graph
has therefore at most $384n$~edges.  We improve this bound by proving
the following theorem.
\begin{theorem}%
    \label{theo:main}%
    A fan-crossing free graph on~$n \geq 3$ vertices has at
    most~$4n-8$ edges.  If the graph has straight edges, it has at
    most $4n-9$ edges.  Both bounds are tight for~$n \geq 10$.
\end{theorem}
A $1$-planar graph is fan-crossing free, so Theorem~\ref{theo:main} generalizes
both Pach and \Toth's and Didimo's bound.  We also extend their lower
bounds by giving tight constructions for every value of~$n$.

In an RAC graph all edges crossed by a given edge~$g$ are orthogonal
to it and therefore parallel to each other, implying that an RAC graph
is fan-crossing free.  Our theorem, therefore, ``nearly'' implies
Didimo \etal's bound: a fan-crossing free graph has at most one edge
more than an RAC graph.

We can completely characterize extremal fan-crossing free graphs, that
is, fan-crossing free graphs on $n$~vertices with $4n-8$ edges: Any
such graph consists of a planar graph~$H$ where each face is a
quadrilateral, together with both diagonals for each face.  This
implies the same properties obtained by Eades and Liotta for extremal
RAC graphs: An extremal fan-crossing free graph is $1$-planar, and is
the union of two maximal planar graphs.

For $k$-fan-crossing free graphs with $k \geq 3$, we obtain the
following result.
\begin{theorem}
    \label{theo:k-main}%
    A $k$-fan-crossing free graph on~$n \geq 3$ vertices has at
    most~$3(k-1)(n-2)$ edges, for $k \geq 3$.
\end{theorem}
This bound is not tight, and the best lower-bound construction we are
aware of has only about $kn$~edges.

Most of the graph families discussed above have a common pattern: the
subgraphs obtained by taking the edges crossed by a given edge~$\edge$
may not contain some forbidden subgraph. We can formalize this notion
as follows: For a topological graph~$\Graph$ and an edge $\edge$
of~$\Graph$, let~$\Graph_\edge$ denote the subgraph of $\Graph$
containing exactly those edges that cross~$\edge$.

A graph property~$\grapro$ is called \emphi{monotone} if it is
preserved under edge-deletions.  In other words, if $\Graph$ has
$\grapro$ and $\Graph'$ is obtained from $\Graph$ by deleting edges,
then $\Graph'$ must have~$\grapro$.  Given a monotone graph
property~$\grapro$, we define a \emphi{derived graph property}
$\edgepro$ as follows: A topological graph $\Graph$ has $\edgepro$ if
for every edge~$\edge$ of~$\Graph$ the subgraph~$\Graph_\edge$
has~$\grapro$.  Some examples are:
\begin{compactitem}
    \item If $\grapro$ is the property that a graph does not contain a
    path of length two, then $\edgepro$ is the property of being
    fan-crossing free;
    \item if $\grapro$ is the property of having at most $k$ edges,
    then $\edgepro$ is $k$-planarity;
    \item if $\grapro$ is planarity, then $\edgepro$ is
    3-quasi-planarity.
\end{compactitem}
We can consider $\edgepro$ for other interesting properties~$\grapro$,
such as not containing a path of length~$k$, or not containing
a~$K_{2,2}$.

We prove the following very general theorem:
\begin{theorem}
    \label{theo:gmain}%
    Let $\grapro$ be a monotone graph property such that any graph on
    $n$ vertices that has~$\grapro$ has at most $O(n^{1+\alpha})$
    edges, for a constant $0 \leq \alpha \leq 1$.  Let $\Graph$ be a
    graph on $n$ vertices that has $\edgepro$.  If $\alpha > 0$, then
    $\Graph$ has $O(n^{1+\alpha})$ edges.  If $\alpha = 0$, then
    $\Graph$ has $O(n \log^{2} n)$ edges.
\end{theorem}
This immediately covers many interesting cases. For instance, a graph
where no edge crosses a path of length~$k$, for a constant~$k$, has at
most $O(n \log^{2} n)$ edges.  Graphs where no edge crosses a
$K_{2,2}$ have at most $\Theta(n^{3/2})$ edges (and this is tight, as
there are graphs with $\Theta(n^{3/2})$ edges that do not contain
a~$K_{2,2}$, implying that no edge can cross a~$K_{2,2}$).

\paragraph{Paper organization.}

Section~\ref{sec:combi} tackles the problem in the simplest settings
involving a single simply-connected ``face'' of a fan-crossing free
graph.  Section~\ref{sec:upperbound} extends this argument to the
whole graph.  Section~\ref{sec:lowerbound} describes lower bound
constructions, and the straight-line case.  In
Section~\ref{sec:k-geq-3}, the argument is extended to the $k$-fan
case.  Section~\ref{sec:general} proves the bound for the case where
we consider the forbidden structure to be a hereditary property
defined on the intersection graphs induced by the edges. Finally,
Section~\ref{sec:conclusions} ends the paper with a discussion and
some open problems.

\section{A combinatorial puzzle: Arrows and fans}
\label{sec:combi}

At the core of our bound lies a combinatorial question that we can
express as follows: An \emphi{$m$-star} is a regular $m$-gon~$\face$
with a set of \emphi{arrows}.  An arrow is a ray starting at a vertex
of~$\face$, pointing into the interior of~$\face$, and exiting through
an edge of~$\face$.

\begin{figure}[t]
  \centerline{\hfill\includegraphics{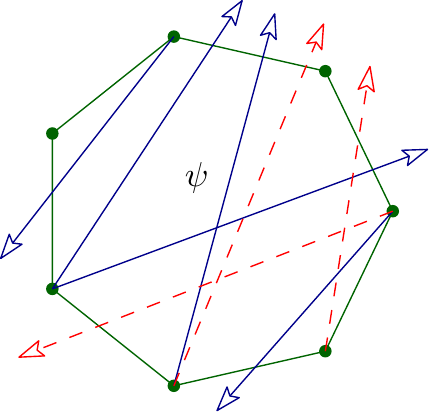}\hfill
    \includegraphics{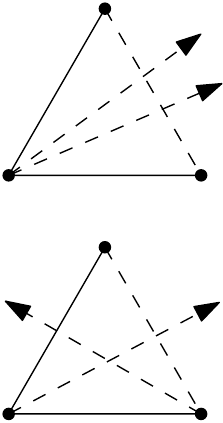}\hfill}
  \caption{Left: A $7$-star. Right: A $3$-star has at most one arrow.}
  \label{fig:mstar}
\end{figure}
We require the set of edges and arrows to be \emphi{fan-crossing
  free}---that is, no edge or arrow intersects two arrows or an edge
and an arrow incident to the same vertex.  The left side of
Figure~\ref{fig:mstar} shows a $7$-star. The dashed arrows are
impossible---each of them forms a fan crossing with the solid edges
and arrows.

The question is: \emphi{How many arrows can an $m$-star possess?}
\begin{observation}
  \label{observation:3star}
  A $3$-star has at most one arrow.
\end{observation}
\begin{proof}
  An arrow from a vertex has to exit the triangle through the opposing
  edge, so no vertex has two arrows. But two arrows from different
  vertices will also form a fan crossing, see the right side of
  Figure~\ref{fig:mstar}.
\end{proof}
It is not difficult to see that a $4$-star possesses at most~$2$
arrows.  The reader may enjoy constructing $m$-stars with
$2m-6$~arrows, for $m \geq 4$.  We conjecture that this bound is
tight.  In the following, we will only prove a weaker bound that is
sufficient to obtain tight results for fan-crossing free graphs.

While we have posed the question in a geometric setting, it is
important to realize that it is a \emphi{purely combinatorial}
question.  We can represent the $m$-star by writing its sequence of
vertices and indicating when an arrow exits~$\face$.  Whether or not
three edges/arrows form a fan crossing can be determined from the
ordering of their endpoints along the boundary of~$\face$ alone.

Let $C=\vertex_1, \ldots, \vertex_m$ be the sequence of vertices
of~$\face$ in counter-clockwise order, such that the $i$\th boundary
edge of~$\face$ is $\edge_i = \vertex_i \vertex_{i+1}$ (all indices
are modulo~$m$).  Consider an arrow~$\edge$ starting at~$\vertex_i$.
It exits~$\face$ through some edge~$\edge_j$, splitting~$\face$ into
two chains $\vertex_{i+1}\dots\vertex_j$ and
$\vertex_{j+1}\dots\vertex_{i-1}$.  The \emphi{length} of $\edge$ is
the number of vertices on the shorter chain.

We will call an arrow \emphi{short} if it has length one. A
\emphi{long arrow} is an arrow of length larger than one.

\begin{lemma}
  \label{lemma:kgon}
  For $m \geq 4$, an $m$-star~$\face$ has at most $2m - 8$ long
  arrows.
\end{lemma}
\begin{proof}
  The proof is by induction over~$m$.
  
  Any arrow in a $4$-star partitions the boundary into chains of
  length one and length two, and so there are no long arrows, proving
  the claim for $m = 4$.
    
  We suppose now that $m > 4$ and that the claim holds for all $4 \leq
  m' < m$.  We \emph{delete all short arrows}, and let $\MeetB$ denote
  the remaining set of arrows, all of which are now long arrows.  Let
  $\edge$ be an arrow of \emph{shortest length}~$\ell$ in~$\face$.
  Without loss of generality, we assume that $\edge$ starts
  in~$\vertex_1$ and exits through edge $e_{\ell+1} = \vertex_{\ell+1}
  \vertex_{\ell+2}$. Then, the following properties hold
  (see~Figure~\ref{fig:kgon}):
  \begin{figure}[t]
    \centerline{%
      \hfill%
      \includegraphics{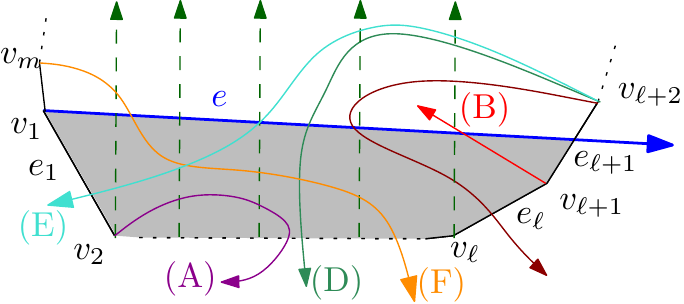}%
      \hfill%
      \includegraphics{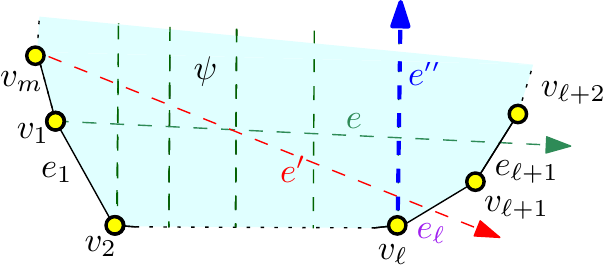}%
      \hfill%
    }
    \caption{Left: properties~(A) to~(F), right: property~(G) for the
      proof of~Lemma~\ref{lemma:kgon}.}
    \label{fig:kgon}
  \end{figure}
    
  \begin{compactenum}[(A)]
  \item Every arrow starting in $\vertex_2,\ldots, \vertex_{\ell+1}$
    must cross $\edge$, as otherwise it would be shorter than $\edge$.
        
  \item There is no arrow that starts in $v_{\ell+1}$. By~(A), such an
    arrow must cross~$\edge$, and so it forms a fan crossing with
    $\edge$ and $\edge_{\ell+1}$.
        
  \item At most one arrow starts in $\vertex_i$, for $i = 2, \ldots,
    \ell$. Indeed, two arrows starting in $\vertex_i$, for $i=2,
    \ldots,\ell$, must cross $\edge$ by~(A), and so they form a fan
    crossing with $\edge$.
        
  \item No arrow starting in~$\vertex_{\ell+2}$ exits through
    $\edge_2, \ldots, \edge_\ell$, as then it would be shorter than
    $\edge$.
        
  \item An arrow starting in~$\vertex_{\ell+2}$ and exiting
    through~$\edge_1$ cannot exist either, as it forms a fan crossing
    with $\edge$ and $\edge_1$.
        
  \item No arrow starting in~$\vertex_m$ crosses $\edge_1, \ldots,
    \edge_{\ell-1}$, as then it would be shorter than $\edge$.
        
  \item The following two arrows cannot both exist: An arrow $\edge'$
    starting in~$\vertex_m$ and exiting through~$\edge_\ell$, and an
    arrow~$\edge''$ starting in~$\vertex_\ell$.  Indeed, if both
    $\edge'$ and $\edge''$ are present, then either~$e''$ exits
    through~$e_{m}$ and forms a fan crossing with $e$ and~$e_{m}$, or
    $e''$ intersects~$e'$ and so $e'$, $e''$, and $e_{\ell}$ form a
    fan crossing (see the right side of~Figure~\ref{fig:kgon}).
  \end{compactenum}
    
  We now create an $(m-\ell+1)$-star $\faceA$ by removing the vertices
  $\vertex_2\dots\vertex_\ell$ with all their incident arrows
  from~$\face$, such that $\vertex_1$ and $\vertex_{\ell+1}$ are
  consecutive on the boundary of~$\faceA$.  An arrow that
  exits~$\face$ through one of the edges $\edge_1\dots\edge_\ell$
  exits~$\faceA$ through the new edge~$g = \vertex_1\vertex_{\ell+1}$.
    
  Let $\MeetB' \subset \MeetB$ be the set of arrows of~$\faceA$, that
  is, the arrows of~$\face$ that do not start
  from~$\vertex_2\dots\vertex_\ell$.  Among the arrows in~$\MeetB'$,
  there are one or two short arrows: the arrow~$\edge$, and the
  arrow~$\edge'$ starting in~$\vertex_m$ and exiting
  through~$\edge_\ell$ in~$\face$ (and therefore through~$g$
  in~$\faceA$) if it exists.  We set $q = 1$ if $\edge'$ exists, and
  else $q=0$.
    
  We delete from~$\faceA$ those one or two short arrows, and claim
  that there is now no fan crossing in~$\faceA$.  Indeed, a fan
  crossing would have to involve the new edge~$g =
  \vertex_1\vertex_{\ell+1}$.  But any arrow that crosses~$g$ must
  also cross~$\edge$, and there is no arrow starting
  in~$\vertex_{\ell+1}$ by~(B).
    
  Since $\ell \geq 2$, we have $m - \ell + 1 < m$, and so by the
  inductive assumption $\faceA$ has at most $2(m - \ell + 1) - 8 = 2m
  - 2\ell -6$ long arrows.  Since there are $1+q$ short arrows
  in~$\MeetB'$, we have $|\MeetB'| \leq 2m - 2\ell - 5 +q $.  By~(C)
  and~(G), we have $|\MeetB| - |\MeetB'| \leq \ell-1 - q$.  It follows
  that
  \begin{align*}
    |\MeetB| & \leq |\MeetB'| + \ell - 1 - q \leq 2m - \ell - 6
    \leq 2m - 8. \qedhere
  \end{align*}
  \aftermathA
\end{proof}

\parpic[l]{\includegraphics{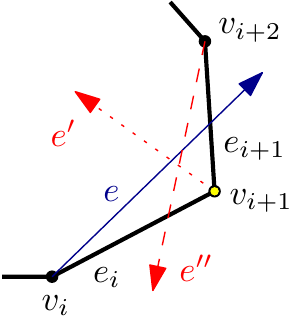}}
It remains to count the short arrows.  Let $e$ be a short arrow, say
starting in~$\vertex_i$ and exiting through~$\edge_{i+1}$. Let us
call~$\vertex_{i+1}$ the \emphi{witness} of~$\edge$.  We observe that
no arrow~$\edge'$ can start in this witness---$\edge'$ would form a
fan crossing with~$\edge$ and~$\edge_{i+1}$.  The
vertex~$\vertex_{i+1}$ can serve as the witness of only one short
arrow: The only other possible short arrow~$\edge''$ with
witness~$\vertex_{i+1}$ starts in~$\vertex_{i+2}$ and exits
through~$\edge_i$.  However, $\edge$, $\edge''$, and~$\edge_i$ form a
fan crossing. \picskip{0}

We can now bound the number of arrows of an~$m$-star.
\begin{lemma}
  \label{lemma:mstar} 
  For $m \geq 3$, an $m$-star~$\face$ has at most $3m - 8$ arrows.
  The bound is attained only for $m = 3$.
\end{lemma}
\begin{proof}
  By Observation~\ref{observation:3star}, the claim is true for $m =
  3$. We consider $m > 3$. By Lemma~\ref{lemma:kgon}, there are at
  most~$2m-8$ long arrows. Each short arrow has a unique witness.  If
  all vertices are witnesses then there is no arrow, and so we can
  assume that at most $m-1$ vertices serve as witnesses, and we have
  at most $m-1$ short arrows, for a total of~$3m-9$ arrows.
\end{proof}

\section{The upper bound for fan-crossing free graphs}
\label{sec:upperbound}

\parpic[l]{\includegraphics{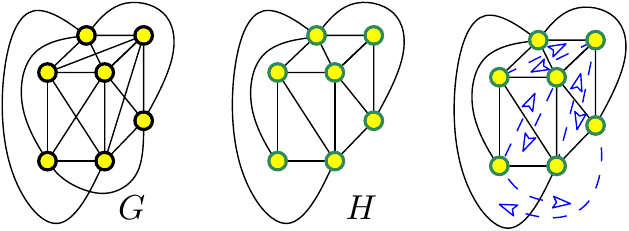}}
\noindent
Let $\Graph=(\Vertices,\Edges)$ be a fan-crossing free graph. We fix
an arbitrary maximal planar subgraph $\GraphA = (\Vertices,\Edges')$
of $\Graph$.  Let~$K = \Edges \setminus \Edges'$ be the set of edges
of~$\Graph$ that is not in~$\GraphA$.  Since $\GraphA$ is maximal,
every edge in~$K$ must cross at least one edge of~$\GraphA$.  We will
replace each edge of~$K$ by \emph{two arrows}.

Let $\edge\in K$ be an edge connecting vertices~$\vertex$
and~$\vertexA$.  The initial segment of~$\edge$ must lie inside a
face~$\face$ of~$\GraphA$ incident to~$\vertex$, the final segment
must lie inside a face~$\faceA$ of~$\GraphA$ incident to~$\vertexA$.
It is possible that $\face = \faceA$, but in that case the
edge~$\edge$ does not entirely lie in the face, as $\GraphA$ is
maximal.  We replace~$\edge$ by two arrows: one arrow starting
in~$\vertex$ and passing through~$\face$ until it exits~$\face$
through some edge; another arrow starting in~$\vertexA$ and passing
through~$\faceA$ until it exits~$\faceA$ through some edge.

In this manner, we replace the set of edges~$K$ by a set of~$2|K|$
arrows.  The result is a planar graph whose faces have been adorned
with arrows. The collection of edges and arrows is fan-crossing free.

Every edge of~$\GraphA$ is incident to two faces of~$\GraphA$, which
can happen to be identical.  If we distinguish the sides of an edge,
the boundary of each face~$\face$ of $\GraphA$ consists of simple
chains of edges.  If $\face$ is bounded, one chain forms the outer
boundary of~$\face$ that makes $\face$~bounded, while all other chains
bound holes inside~$\face$; if $\face$ is unbounded, then all chains
bound holes in~$\face$.

\parpic[l]{\includegraphics{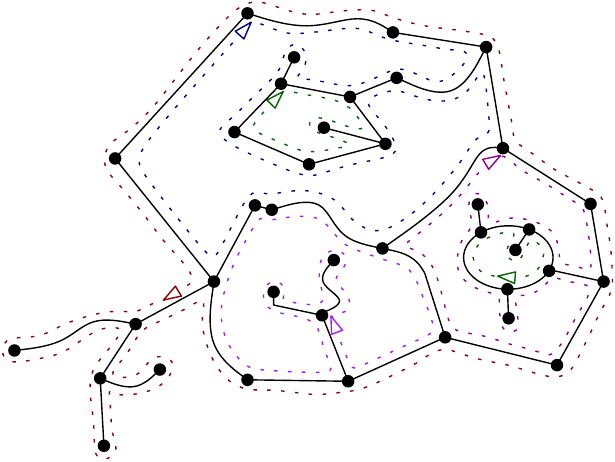}}
If the graph~$\GraphA$ is connected, then the boundary of each face
consists of a single chain.  Let $\face$ be such a face whose boundary
chain consists of $m$~edges (where edges that bound~$\face$ on both
sides are counted twice).  Then $\face$ has at most $3m-8$ arrows.
This follows immediately from Lemma~\ref{lemma:mstar}: Recall that $m$-stars
can be defined purely combinatorially.  Whether three edges form a fan
crossing can be decided solely by the ordering of their endpoints
along the boundary chain.  The boundary of a simply connected face is
a single closed chain, and so Lemma~\ref{lemma:mstar} applies to this setting,
see the side figure.  \picskip{0}

Unfortunately, we cannot guarantee that $\GraphA$ is connected.  The
following lemma bounds the number of arrows of a face~$\face$ in terms
of its complexity and its number of boundary chains.  The complexity
of a face is the total number of edges of all its boundary chains,
where edges that are incident to the face on both sides are counted
twice.
\begin{lemma}
  \label{lemma:face} 
  A face of~$\GraphA$ of complexity~$m$ bounded by~$p$ boundary chains
  possesses at most $3m + 8p - 16$ arrows.  The bound can be attained
  only when $m=3$ and~$p=1$.
\end{lemma}
We will prove the lemma below, but let us first observe how it implies
the upper bound on the number of edges of fan-crossing free graphs.
\begin{lemma}
  \label{lemma:upper} 
  A fan-crossing free graph~$\Graph$ on~$n$ vertices has at most
  $4n-8$ edges.
\end{lemma}
\begin{proof}
  Let $m$ be the number of edges, let $r$ be the number of faces, and
  let $p$ be the number of connected components of~$\GraphA$.  Let
  $\Faces$ be the set of faces of~$\GraphA$.  For a face $\face \in
  \Faces$, let $m(\face)$ denote the complexity of~$\face$, let
  $p(\face)$ denote the number of boundary chains of~$\face$, and
  let~$a(\face)$ denote the number of arrows of~$\face$.
    
  We have $\sum_{\face \in \Faces}m(\face) = 2m$ and
  $\sum_{\face\in\Faces}(p(\face)-1) = p - 1$ (each component is
  counted in its unbounded face, except that we miss one hole in the
  global unbounded face).
    
  The graph~$\Graph$ has $|E| = m + |K|$ edges.  Using
  Lemma~\ref{lemma:face} we have
  \begin{align*}
    2|E| & = 2m + 2|K| =
    \sum_{\face \in \Faces}m(\face) + \sum_{\face \in \Faces} a(\face) \\
    & \leq \sum_{\face \in \Faces} \big(4m(\face) + 8p(\face) - 16\big) \\
    & = 4\sum_{\face\in\Faces}m(\face) + 8\sum_{\face\in\Faces}
    (p(\face) - 1)
    - 8r\\
    & = 8m + 8p - 8 - 8r.  \intertext{By Euler's formula, we have
      $n - m + r = 1 + p$, so $m - r = n - 1 - p$, and we have}
    2|E| & \leq 8(m-r) + 8p - 8 = 8n - 8 - 8p + 8p - 8 = 8n -
    16. \qedhere
  \end{align*}
  \aftermathA
\end{proof}

It remains to fill in the missing proof.

\begin{proof}[Proof of Lemma~\ref{lemma:face}]
  Let $\face$ be a face of~$\GraphA$, and let $m = m(\face)$ and $p =
  p(\face)$ be its complexity and its number of boundary components.
  A boundary component is a chain of edges, and could possibly
  degenerate to a single isolated vertex.
    
  We say that two boundary chains $\xi$ and $\zeta$ are
  \emphi{related} if an arrow starting in a vertex of~$\xi$ ends in an
  edge of~$\zeta$, or vice versa. Consider the undirected graph whose
  nodes are the boundary chains of~$\face$ and whose arcs connect
  boundary chains that are related.  If this graph has more than one
  connected component, we can bound the number of arrows separately
  for each component, and so in the following we can assume that all
  boundary chains are (directly or indirectly) related.
    
  \begin{figure}[ht]
    \centerline{\hfill\includegraphics[page=1]{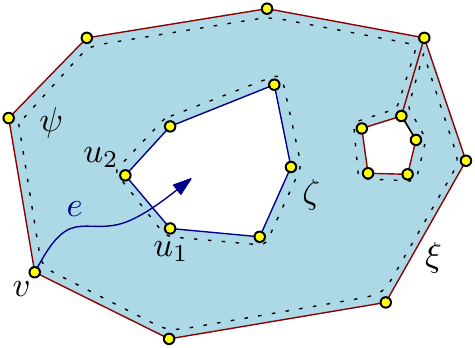}\hfill
      \includegraphics[page=2]{figs/bridge}\hfill}
    \caption{Building a bridge between~$\xi$ and~$\zeta$.}
    \label{fig:bridge}
  \end{figure}
  Consider two related boundary chains~$\xi$ and~$\zeta$.  By
  assumption there must be an arrow~$\edge$, starting at a
  vertex~$\vertex \in\xi$, and ending in an
  edge~$\vertexA_1\vertexA_2$ of~$\zeta$.  We create a new vertex~$z$
  on~$\zeta$ at the intersection point of $\edge$
  and~$\vertexA_1\vertexA_2$, split the boundary edge
  $\vertexA_1\vertexA_2$ into two edges $\vertexA_1z$ and
  $z\vertexA_2$, and insert the two new boundary edges~$\vertex z$
  and~$z \vertex$, see Figure~\ref{fig:bridge}.  This operation has
  increased the complexity of~$\face$ by three.  Note that some arrows
  of $\face$ might be crossing the new boundary edges---these arrows
  will now be shortened, and end on the new boundary edge.
    
  The two boundary chains~$\xi$ and~$\zeta$ have now merged into a
  single boundary chain.  In effect, we have turned an arrow into a
  ``bridge'' connecting two boundary chains.  No fan crossing is
  created, since all edges and arrows already existed. We do create a
  new vertex~$z$, but no arrow starts in~$z$, and so this vertex
  cannot cause a fan crossing.
    
  We insert $p-1$ bridges in total and connect all~$p$ boundary
  chains.  In this manner, we end up with a face~$\faceA$ whose
  boundary is a single chain consisting of $m' = m + 3(p - 1)$ edges.
    
  If $m' = 3$, then $\faceA$ has at most one arrow, by
  Observation~\ref{observation:3star}.  This case happens only for $m
  = 3$ and $p = 1$, and is the only case where the bound is tight.
    
  If $m' > 3$, then we can apply~Lemma~\ref{lemma:kgon} to argue
  that~$\faceA$ has at most~$2m'- 8$ long arrows.  To count the short
  arrows, we observe that the vertex~$z$ created in the
  bridge-building process cannot be the witness of a short arrow: such
  a short arrow would imply a fan crossing in the original
  face~$\face$.  It is also not the starting point of any arrow.
  
  It follows that building a bridge increases the number of possible
  witnesses by only one (the vertex~$\vertex$ now appears twice on the
  boundary chain).  There are thus at most $m + p - 1$ possible
  witnesses in~$\faceA$.  However, if all of these vertices are
  witnesses, then there is no arrow at all, and so there are at most
  $m + p - 2$ short arrows.
    
  Finally, we converted $p - 1$ arrows of~$\face$ into bridges to
  create~$\faceA$.  The total number of arrows of~$\face$ is therefore
  at most
  \begin{align*}
    2m' - 8 + (m+p-2) + (p-1)%
    &=%
    2(m + 3p - 3) - 8 + (m + p - 2) + (p - 1) \\
    & =%
    3m + 8p - 17. \qedhere
  \end{align*}
  \aftermathA
  \aftermathA
\end{proof}

\section{Lower bounds and the straight-line case}
\label{sec:lowerbound}

Consider a fan-crossing free graph~$\Graph$ with $4n-8$~edges.  This
means that the inequality in the proof of Lemma~\ref{lemma:upper} must
be an equality.  In particular, every face~$\face$ of~$\GraphA$ must
have exactly~$3m(\face) + 8p(\face) - 16$ arrows.  By
Lemma~\ref{lemma:face}, this is only possible if $\face$ is a
triangle, and so we have proven
\begin{lemma}
  \label{lemma:maximal} 
  A fan-crossing free graph~$\Graph$ with $4n-8$ edges contains a
  planar triangulation~$\GraphA$ of its vertex set.  Each triangle
  of~$\GraphA$ possesses exactly one arrow.
\end{lemma}
Note that the arrow must necessarily connect a vertex of the triangle
with the opposite vertex in the triangle adjacent along the opposite
edge, as otherwise, it forms a fan crossing, and so we get the
following second characterization of extremal fan-crossing free
graphs:
\begin{lemma}
  \label{lemma:maximal-quad} 
  A fan-crossing free graph~$\Graph$ with $4n-8$ edges contains a
  planar graph~$Q$ on its vertex set, where each face of~$Q$ is a
  quadrilateral. $\Graph$ is obtained from~$Q$ by adding both
  diagonals for each face of~$Q$.
\end{lemma}
By Euler's formula, a planar graph~$Q$ on~$n$ vertices whose faces are
quadrilaterals has $2n-4$ edges and $n-2$ faces.  Adding both
diagonals to each face of $Q$, we obtain a fan-crossing free
graph~$\Graph$ with $4n-8$ edges.  However, it turns out that this
construction needs to be done carefully: Otherwise, diagonals of two
distinct faces of~$Q$ can connect the same two vertices, and the
result is not a simple graph.  And indeed, it turns out that for $n\in
\{7, 9\}$, no (simple) fan-crossing free graph with $4n-8$ edges
exists!  When $n \geq 8$ is a multiple of four,
Figure~\ref{fig:quadrilaterals} shows planar graphs~$Q$ where every
bounded face is a convex quadrilateral.  Since all their diagonals are
straight, no multiple edges can arise.  Only the two diagonals of the
unbounded face are not straight and need to be checked individually.
We will return to other values of~$n$ below.
\begin{figure}[ht]
  \centerline{\includegraphics[page=1]{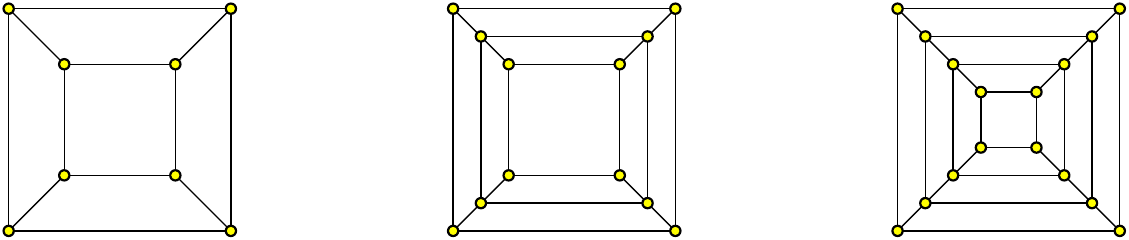}}
  \caption{Planar graphs with convex quadrilateral faces for $n \in
    \{8, 12, 16\ldots\}$.}
  \label{fig:quadrilaterals}
\end{figure}

Lemma~\ref{lemma:maximal} implies immediately that a fan-crossing free
graph with $4n-8$ edges cannot exist if \emph{all} edges are straight:
Since the unbounded face of~$\GraphA$ is a triangle, it cannot contain
a straight arrow, and so any fan-crossing free graph drawn with
straight edges has at most $4n-9$ edges.  This bound is tight: for
any~$n \geq 6$, we can construct a planar graph~$Q$ such that two
faces of~$Q$ are triangles, and all other faces are convex
quadrilaterals.  Euler's formula implies that $Q$ has $2n-3$ edges and
$n-1$ faces, and adding both diagonals to each quadrilateral face
results in a graph with $2n-3 + 2(n-3) = 4n-9$ edges.  The
construction of~$Q$ is shown in Figure~\ref{fig:triangles}.  The upper
row shows the construction when $n \geq 6$ is a multiple of three.  If
$n \equiv 1 \pmod{3}$, we replace the two innermost triangles as shown
in the lower left of the figure. If $n \equiv 2 \pmod{3}$, the two
innermost triangles are replaced as in the lower right figure.
\begin{lemma}
  \label{lemma:straight} 
  A fan-crossing free graph drawn with straight edges has at most
  $4n-9$ edges. This bound is tight for $n\geq6$.
\end{lemma}
\begin{figure}[ht]
  \centerline{\includegraphics[page=2]{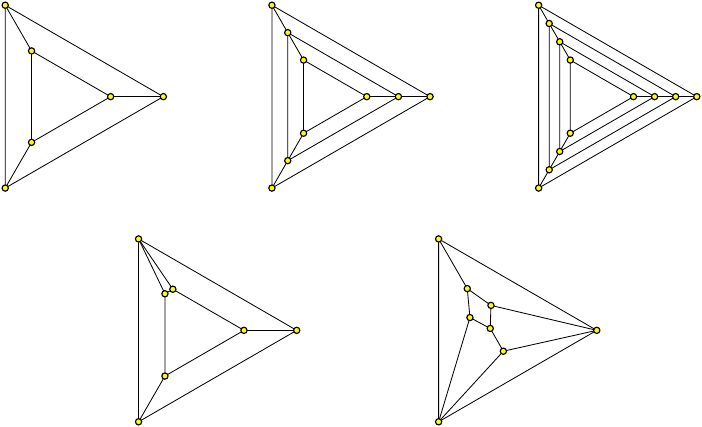}}
  \caption{Adding both diagonals for each quadrilateral face results
    in a straight-line fan-crossing free graph with $4n-9$ edges.}
  \label{fig:triangles}
\end{figure}

We observe that the extremal fan-crossing free graph constructed
for~$n=6$ has $15$~edges and is therefore the complete graph~$K_6$. It
follows that the complete graph~$K_n$ is fan-crossing free for~$n \leq
6$, and it remains to discuss extremal fan-crossing free graphs for $n
\geq 7$ when the edges are not necessarily straight.
\begin{lemma}
  Extremal fan-crossing free graphs with $4n-8$ edges exist for $n =
  8$ and all $n \geq 10$.  For $n \in \{7, 9\}$, extremal fan-crossing
  free graphs have $4n-9$ edges.
\end{lemma}
\begin{proof}
  Let $n \in \{7, 9\}$, and assume $\Graph$ is a fan-crossing free
  graph on~$n$ vertices with~$4n-8$ edges.  By
  Lemma~\ref{lemma:maximal-quad}, $\Graph$ contains a planar graph~$Q$
  whose faces are quadrilaterals.  We claim that $Q$ has a vertex~$v$
  of degree two.  This implies that the two quadrilaterals incident
  to~$v$ have diagonals connecting the same two vertices, a
  contradiction.  It follows that a fan-crossing free graph with
  $4n-8$ edges does not exist.  By Lemma~\ref{lemma:straight}, a
  fan-crossing free graph with $4n-9$ edges does exist.
    
  For $n = 7$, the claim is immediate: $Q$ has $2n-4=10$ edges, and so
  its average degree is $20/7 < 3$, and a vertex of degree two must
  exist.
    
  For $n = 9$, the total degree of~$Q$ is $2(2n-4) = 28$.  We assume
  there is no vertex of degree two. Since nine vertices of degree
  three already contribute~$27$ to the total degree, it follows that
  there is one vertex~$w$ of degree four, while the other eight
  vertices have degree three.  Since $Q$ contains only even cycles, it
  is a bipartite graph, and the vertices can be partitioned into two
  classes~$A$ and~$B$.  Let $w \in A$, and let $k = |A| - 1$.  The
  total degree of vertices in~$A$ and in~$B$ is identical, so we have
  $4 + 3k = 3(8-k)$, or $6k=20$, a contradiction.
    
  The case $n = 8$ was already handled in
  Figure~\ref{fig:quadrilaterals}, so consider $n \geq 10$.  We will
  again start with a planar graph~$Q$ with quadrilateral faces.  To
  avoid diagonals connecting identical vertices, we would like to make
  \emph{all} faces convex.  This is obviously not possible when
  drawing the graph in the plane, and so we will draw the graph on the
  sphere, such that all faces are spherically-convex quadrilaterals.
    
  For even~$n$, we place a vertex at the North pole and at the South
  pole each.  The remaining $n-2$ vertices form a zig-zag chain near
  the equator, distributing the points equally on two circles of equal
  latitude, see left-hand side of Figure~\ref{fig:spherical}
  (for~$n=10$).  For odd~$n$, we place two vertices close to the North
  pole, one vertex at the South pole, and let the remaining~$n-3$
  vertices again form a zig-zag chain near the equator, see the
  right-hand side of Figure~\ref{fig:spherical} (for $n = 11$).  One
  of the resulting quadrilaterals is long and skinny, but it is still
  spherically convex.
  \begin{figure}[ht]
    \centerline{\includegraphics{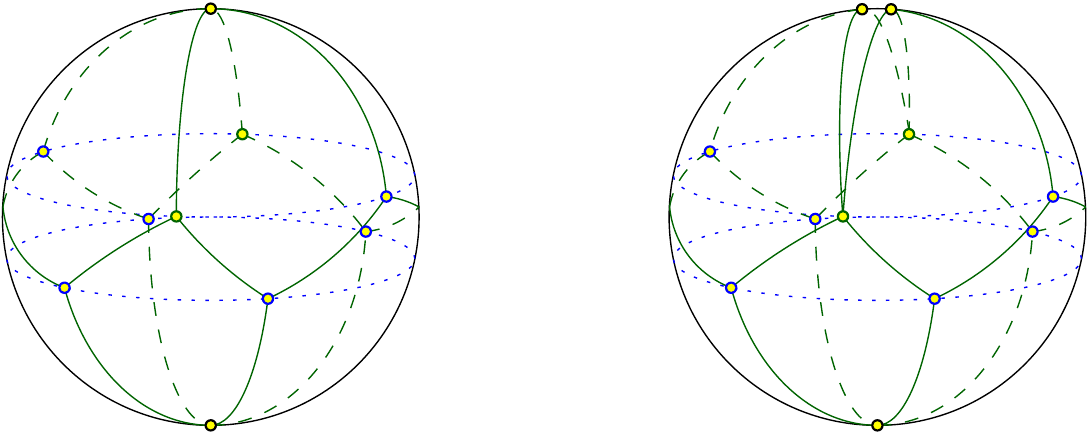}}
    \caption{Spherically-convex quadrilateralization for even~$n$
      and odd~$n$.}
    \label{fig:spherical}
  \end{figure}
\end{proof}

\section[k-fan-crossing free graphs]{$k$-fan-crossing free
  graphs for $k \geq 3$}
\label{sec:k-geq-3}

Our proof of Theorem~\ref{theo:k-main} has the same structure as for
the case~$k = 2$: Let $G = (V, E)$ be a $k$-fan-crossing free graph
for~$k \geq 3$.  We fix an arbitrary maximal planar subgraph~$H = (V,
E')$ of~$G$, and let $K = E \setminus E'$.  We replace each edge
in~$K$ by two arrows, so we end up with a planar graph~$H$ whose faces
have been adorned with $2|K|$~arrows in total.
\begin{lemma}
  \label{lemma:k-face} A $k$-fan-crossing free face of~$H$ of
  complexity~$m \geq 3$ bounded by~$p$ boundary chains possesses at
  most $3(k-1)(m + 2p - 4) - 2m + 3$ arrows, for $k \geq 3$.
\end{lemma}
As in the case~$k=2$, we prove Lemma~\ref{lemma:k-face} by converting arrows
connecting different boundary components of the face into bridges,
until we have obtained a single boundary chain.  The technical details
are more complicated, and so we first discuss how Lemma~\ref{lemma:k-face}
implies Theorem~\ref{theo:k-main}.

\begin{proof}[Proof of Theorem~\ref{theo:k-main}]
  Let $m$ be the number of edges, let $r$ be the number of faces, and
  let $p$ be the number of connected components of~$\GraphA$.  Let
  $\Faces$ be the set of faces of~$\GraphA$.  For a face $\face \in
  \Faces$, let $m(\face)$ denote the complexity of~$\face$, let
  $p(\face)$ denote the number of boundary chains of~$\face$, and
  let~$a(\face)$ denote the number of arrows of~$\face$. As observed
  before, we have $\sum_{\face \in \Faces}m(\face) = 2m$ and
  $\sum_{\face\in\Faces}(p(\face)-1) = p - 1$.
    
  The graph~$\Graph$ has $|E| = m + |K|$ edges.  Since $m(\face) \geq
  3$, Lemma~\ref{lemma:k-face} implies $m(\face) + a(\face) \leq
  3(k-1)(m(\face) + 2p(\face)-4)$, and so
  \begin{align*}
    2|E| & = 2m + 2|K| = \sum_{\face \in \Faces}
    m(\face) + \sum_{\face \in \Faces} a(\face) \\\
    & \leq \sum_{\face \in \Faces} \Big((3k-3)\big(m(\face) +
    2(p(\face) - 1) - 2\big)\Big) \\
    & = 3(k-1)\Big(\sum_{\face\in\Faces}m(\face) +
    2\sum_{\face\in\Faces}
    (p(\face) - 1)- 2r\Big)\\
    & = 6(k-1)(m + (p-1) - r) = 6(k-1)(n-2).
  \end{align*}
  In the last line we used Euler's formula $n - m + r = 1 + p$
  for~$H$.
\end{proof}

To prove Lemma~\ref{lemma:k-face}, we need a $k$-fan equivalent of
Lemma~\ref{lemma:mstar}.  One additional difficulty is now keeping the
contribution of the additional vertices formed by creating bridges
low.  We need a few new definitions.

Again, an \emphi{$m$-star} is a regular $m$-gon~$\face$ with a set of
arrows.  We require the set of edges and arrows to be $k$-fan-crossing
free.  We number the vertices counter-clockwise as~$v_1, v_2, \dots,
v_m$, and denote the edge between~$v_i$ and~$v_{i+1}$ as~$e_i$ (where
index arithmetic is again modulo~$m$).  We let $a_{ij}$ denote the
number of arrows starting in vertex~$v_i$ and exiting through
edge~$e_j$.  The \emphi{degree}~$a_i$ of vertex~$v_i$ is the total
number of arrows starting in~$v_i$.  We define the length of an arrow
and the notions of short and long arrows as before.

We need to distinguish different kinds of vertices. If vertex~$v_i$
has degree~$a_i = 0$, we call it \emphi{void}.  If, in addition, its
left neighbor~$v_{i-1}$ has no short arrow passing over~$v_i$, that
is, if $a_i = 0$ and $a_{i-1,i} = 0$, then we call it
\emphi{left-light}.  Similarly, if $a_i = 0$ and $a_{i+1,i-1} = 0$,
then we call~$v_i$ \emphi{right-light}.  A vertex is \emphi{light}
when it is left-light or right-light.  A vertex is \emphi{heavy} when
its degree is non-zero.  Finally, we do not allow a left-light vertex
to be adjacent to a right-light vertex.  What this means is that if a
sequence of consecutive vertices, say~$v_1,v_2,\dots,v_t$ is a maximal
sequence of zero-degree vertices, then either they are all left-light
(if $a_{m,1} = 0$), or they are all right-light (if~$a_{t+1,t-1} =
0$), or only~$t-1$ of them are light and the last one is counted as
void.

We can now formulate our lemma.
\begin{lemma}
  \label{lemma:k-mstar}
  For $m \geq 3$, $k \geq 3$, and $h\geq 2$, an $m$-star~$\face$ with
  $h$~heavy vertices, $\lambda$~light vertices, and~$\nu$~void (but
  not light) vertices has at most $(3k-5)h + k\lambda + (2k-3)\nu -
  (6k-9)$ arrows.
\end{lemma}
Before we prove Lemma~\ref{lemma:k-mstar}, let us first see how it
implies Lemma~\ref{lemma:k-face}.
\begin{proof}[Proof of Lemma~\ref{lemma:k-face}]
  We consider a face~$\face$ of~$\GraphA$ of complexity~$m$ bounded by
  $p$~boundary components.  As in the proof of Lemma~\ref{lemma:face},
  we can assume that all boundary components are related.  Thus there
  must be boundary components~$\xi, \zeta$ and an arrow starting at a
  vertex~$v \in \xi$ and ending in an edge~$u_1u_2$ of~$\zeta$.  More
  precisely, there could be at most~$k-1$ such arrows starting in~$v$
  and ending in~$u_1u_2$.  We pick the two extreme arrows, that is,
  the one closest to~$u_1$ and the one closest to~$u_2$, and convert
  them to bridges, creating two new vertices~$z_1, z_2$ on $u_1u_2$,
  see Figure~\ref{fig:bridge_radial}.
  \begin{figure}[ht]
    \centerline{\hfill\includegraphics[page=1]{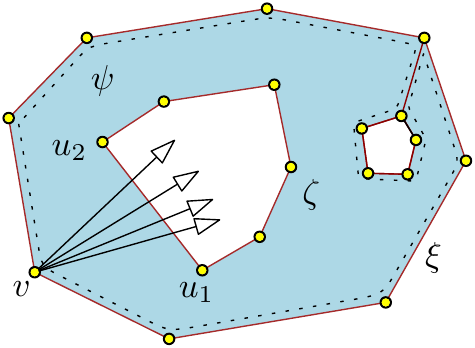}\hfill\
      \includegraphics[page=2]{figs/bridge_radial}\hfill}
    \caption{Building a bridge between~$\xi$ and~$\zeta$.}
    \label{fig:bridge_radial}
  \end{figure}
    
  Inserting these two bridges increases the complexity of the face by
  three.  However, $z_1$ and~$z_2$ are by construction light vertices
  (in Figure~\ref{fig:bridge_radial}, $z_1$~is right-light, $z_2$~is
  left-light). There are at most $k-1$ arrows from~$v$ to~$u_1u_2$
  that are deleted from the face.
    
  We continue in this manner until the face is bounded by a single
  boundary chain.  Every bridge adds one heavy vertex and two light
  vertices to the complexity of the face. (Note that later bridges
  could possibly create vertices on bridges built earlier, but that
  doesn't stop the vertices from being light.)  We finally obtain a
  face with $m + p-1 \geq 2$ heavy vertices and $2(p-1)$ light
  vertices, and during the process we deleted at most $(k-1)(p-1)$
  arrows.  Applying Lemma~\ref{lemma:k-mstar}, the total number of
  arrows of~$\face$ is at most
  \begin{align*}
    & \hspace{-1cm}%
    (3k-5)(m + p - 1) + 2k(p-1) - (6k-9) + (k-1)(p-1) \\
    & = (3k-5)m + (3k - 5 + 2k + k - 1)(p-1) - (6k-9) \\
    & = (3k-3)m - 2m + (6k-6)(p-1) - (6k-6) + 3 \\
    & = 3(k-1)(m + 2p - 4) - 2m + 3. \qedhere
  \end{align*}
  \aftermathA
\end{proof}

Finally, it remains to prove the bound on the number of arrows of a
$k$-fan-crossing free $m$-star.
\begin{proof}[Proof of Lemma~\ref{lemma:k-mstar}]
  Let $A(h, \lambda, \nu)$ denote the maximum number of arrows of an
  $m$-star with $h$~heavy vertices, $\lambda$~light vertices, and
  $\nu$~void but not light vertices, where $m = h + \lambda + \nu$.
  We set
  \[
  B(h, \lambda, \nu) = (3k-5)h + k\lambda + (2k-3)\nu - (6k-9),
  \]
  and show by induction over~$m$ that $A(h,\lambda, \nu) \leq B(h,
  \lambda, \nu)$ under the assumption that $m \geq 3$, $k \geq 3$, $h
  \geq 2$.
    
  We have several base cases. The reader may enjoy verifying that the
  claim holds for the triangle and the quadrilateral:
  \begin{align*}
    A(3, 0, 0) & = 3k-6 = B(3, 0, 0), \\
    A(2, 0, 1) & = 2k -4 = B(2, 0, 1), \\
    A(2, 1, 0) & = k - 1 = B(2, 1, 0), \\
    A(4, 0, 0) & = 5k - 9 \leq 6k - 11 = B(4, 0, 0), \\
    A(3, 0, 1) & = 4k - 6 \leq 5k - 9 = B(3, 0, 1), \\
    A(3, 1, 0) & = 3k - 5 \leq 4k - 6 = B(3, 1, 0), \\
    A(2, 0, 2) & = 4k-8 \leq 4k-7 = B(2,0,2), \\
    A(2, 1, 1) & = 3k-5 \leq 3k-4 = B(2,1,1), \text{ and } \\
    A(2, 2, 0) & = 2k-2 \leq 2k-1 = B(2,2,0).
  \end{align*}
    
  The second base case is when $m > 4$ and all vertex degrees are at
  most~$k-2$.  In this case we have $A(h, \lambda, \nu) \leq (k-2)h$.
  If $h \geq 3$ then
  \begin{align*}
    (3k-5)h - (6k-9) - (k-2)h = (2k-3)(h-3) \geq 0,
  \end{align*}
  and the claim holds.  If $h = 2$ then $\lambda+\nu > 2$ and so
  \begin{align*}
    A(2, \lambda, \nu) & \leq 2k - 4 \leq 2k-1 = B(2, 2, 0) \leq
    B(2, \lambda, \nu).
  \end{align*}
    
  The third base case is when $m > 4$ and $h = 2$, and the two heavy
  vertices are adjacent.  Let $v_1$ and $v_2$ denote the two heavy
  vertices. No arrows start from the other
  vertices~$v_3,\ldots,v_m$. If no arrows starting from $v_1$ and
  $v_2$ intersect each other, then
  \begin{align*}
    A(2,\lambda,\nu) & \leq (k-1)(m-2) = (k-1)\lambda+(k-1)\nu
    \leq B(2,\lambda,\nu).
  \end{align*}
  Assume now that there is at least one arrow from $v_1$ that
  intersects an arrow from~$v_2$. Let $x_1$ denote the rightmost arrow
  from~$v_1$ and let $x_i$ (for $i > 1$) denote the $i$\th arrow from
  $x_1$ in counter-clockwise order.  Similarly, let $y_1$ denote the
  leftmost arrow from~$v_2$ and let $y_i$ (for $i > 1$) denote the
  $i$\th arrow from~$v_2$ in clockwise order. We claim that only the
  arrows $x_1,\ldots,x_{k-1}$ from $v_1$ can intersect the arrows
  $y_1,\ldots,y_{k-1}$ from $v_2$. Indeed, $x_1$ cannot intersect more
  than $k-1$ arrows starting from $v_2$, since otherwise, it forms a
  $k$-fan crossing. If $x_1$ intersects $k-1$ arrows from $v_2$, it
  has to intersect $y_1,\ldots,y_{k-1}$. Since $x_1$ is the rightmost
  arrow from $v_1$, $y_j$ for $j \geq k$ must lie on the right side
  of~$x_1$, which implies that $y_j$ cannot intersect any arrow
  from~$v_1$. Similarly, since $y_1$ is the leftmost arrow from $v_2$,
  $x_j$ for $j \geq k$ cannot intersect any arrow from $v_2$, proving
  the claim. This implies that there can exist at most $(k-1)(m-2) +
  (k-1) = (k-1)(m-1)$ arrows in total, and so
  \begin{align*}
    & \hspace{-1cm}%
    [(3k-5)(2+\nu+\lambda)+(k-1)\nu+(2k-3)\lambda-(6k-9)] -
    (k-1)(1+\nu+\lambda)  \\
    & = (2k-4)(1+\nu+\lambda) + (k-1)\nu+(2k-3)\lambda - (3k-4)
    \geq 0,
  \end{align*}
  for $k \geq 3$ and $\nu+\lambda \geq 3$, and thus $A(2,\nu,\lambda)
  \leq B(2,\nu,\lambda)$.
    
  We now turn to the inductive step. Consider an $m$-star~$\face$ with
  $A(h, \lambda, \nu)$ arrows, where $m > 4$, $h \geq 2$, at least one
  vertex has degree at least~$k-1$, and if $h = 2$ then the two heavy
  vertices are not adjacent.
    
  We consider first the case where there are no long arrows: all
  arrows are short.  By renumbering, we can assume that $v_2$ has
  degree $a_2 \geq k-1$. This implies $a_{13} = 0$ and $a_{31} = 0$,
  see Figure~\ref{fig:radial_free_short}.
  \begin{figure}[t]
    \centerline{\includegraphics{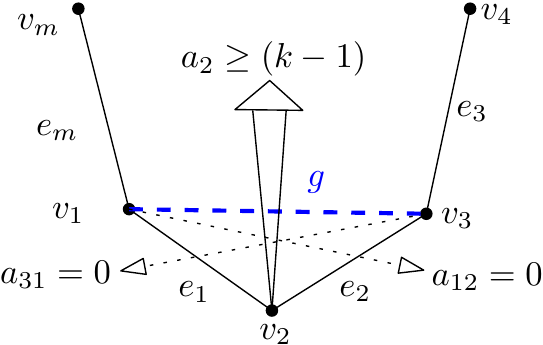}}
    \caption{When all arrows are short.}
    \label{fig:radial_free_short}
  \end{figure}
  We now construct an $(m-1)$-star~$\faceA$ by deleting~$v_2$ and its
  arrows, and inserting the edge~$g = v_1v_3$.  Since $\face$ has only
  short arrows, the only arrows intersecting $g$ are~$a_{m1} \leq k-1$
  arrows starting in~$v_m$ and $a_{42} \leq k-1$ arrows starting
  in~$v_4$.  Since $m > 4$ the vertices $v_m$ and $v_4$ are distinct,
  and this implies that $\faceA$ has no $k$-fan-crossing.  The
  vertex~$v_2$ is obviously heavy and has $a_2 \leq 2(k-1)$.  If $v_3$
  is left-light in~$\face$, it is still left-light in~$\faceA$
  since~$a_{13} = 0$.  Similarly, if $v_1$ is right-light in~$\face$,
  it remains so in~$\faceA$.  (It cannot happen that $v_1$ is
  right-light and $v_3$~is left-light in~$\face$, as then $a_2 = 0 <
  k-1$.)  We thus have
  \begin{align*}
    A(h, \lambda, \nu) & \leq 2k-2 + B(h-1, \lambda, \nu) \leq
    B(h, \lambda, \nu).
  \end{align*}
  
  It remains to consider the case where $\face$ has long arrows.  Let
  $e$ denote the \emph{shortest} long arrow in~$\face$ and let $\ell$
  denote its length. Renumbering the vertices we can assume that
  $e$~starts in~$v_1$ and exits through edge~$e_{\ell+1} =
  v_{\ell+1}v_{\ell+2}$. See Figure~\ref{fig:radial_free_long}.
  \begin{figure}[t]
    \centerline{\includegraphics{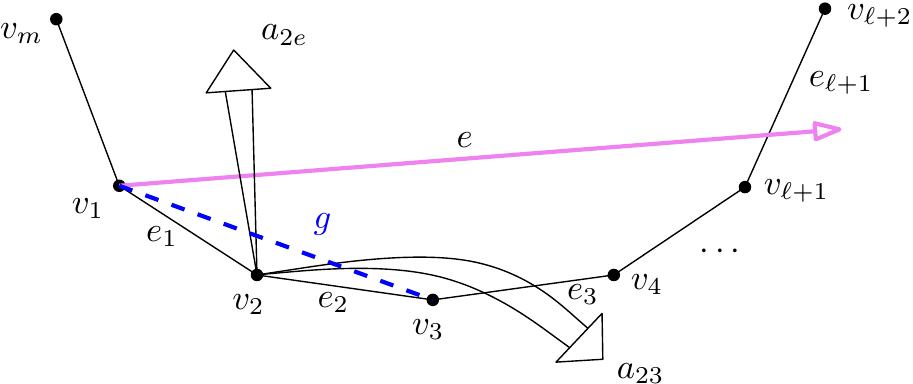}}
    \caption{$e$ is the shortest long arrow.}
    \label{fig:radial_free_long}
  \end{figure}
  The following property holds:
  \begin{compactenum}
  \item[(A)] Every long arrow starting in
    $v_2,\ldots,v_{\ell+1}$ must cross $e$, as otherwise it
    would be shorter than $e$.
  \end{compactenum}
  Let $a_{2e}$ denote the number of arrows starting in~$v_2$ that
  cross~$e$. These arrows cannot form a $k$-fan, so we have $a_{2e}
  \leq k-1$.  An arrow starting in~$v_2$ that does not cross $e$ is
  short by~(A) and must exit through~$e_3$. This implies that $a_{2} =
  a_{2e} + a_{23} \leq 2(k-1)$.
    
  We now create an $(m-1)$-star~$\faceA$ by deleting~$v_2$ and all its
  arrows, deleting the $a_{12}$ arrows starting in~$v_1$ and exiting
  through~$e_3$, deleting the $a_{31}$ arrows starting in~$v_3$ and
  exiting through~$e_1$, and inserting the new edge~$g=v_1v_3$.
    
  If $\faceA$ has a $k$-fan crossing, it must involve the new
  edge~$g$. There are three ways in which this could happen:
  \begin{compactenum}[\qquad \bf {Case}~(a):]
  \item $g$ and $k-1$ arrows starting in
    $v_1$ are intersected by an arrow~$x$ that also
    intersects~$e_2$. See left side of Figure~\ref{fig:create-fan}.
    
  \item $g$ and $k-1$ arrows starting in $v_3$ are intersected by an
    arrow~$y$ that also intersects $e_1$. See right side of
    Figure~\ref{fig:create-fan}.
    
  \item $g$ intersects $k$ arrows starting in the same vertex.
  \end{compactenum}

  \begin{figure}[t]
    \centerline{%
      \hfill%
      \includegraphics{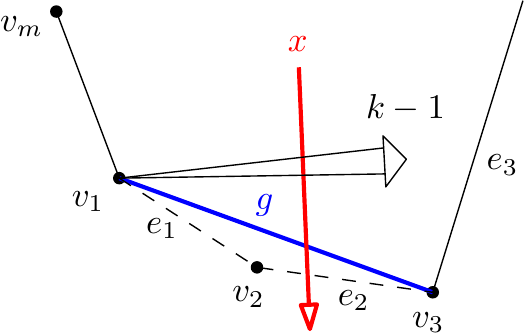}%
      \hfill%
      \includegraphics{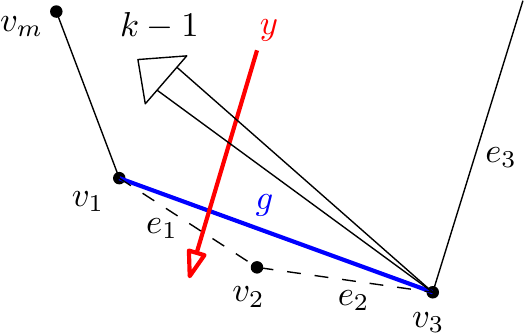}%
      \hfill%
    }
    \caption{Two cases where $g$ creates a $k$-fan in~$\faceA$.}
    \label{fig:create-fan}
  \end{figure}
  We first observe that case~(c) cannot happen: If the $k$ arrows
  intersecting $g$ also intersect~$e$, they form a $k$-fan crossing in
  $\face$. If there is an arrow that does not intersect~$e$, it must
  start in one of $v_4,\ldots,v_{\ell+1}$, and so it is short
  by~(A). To intersect $g$, these arrows must start in $v_4$ and exit
  through $e_2$, and there are at most $(k-1)$ such arrows.
    
  If case~(a) happens, we delete from~$\faceA$ one more arrow, namely,
  the ``lowest'' arrow starting in~$v_1$. If case~(b) happens, we take
  out the ``lowest'' arrow starting in $v_3$.  This ensures that
  $\faceA$ has no $k$-fan crossing. The inductive assumption holds
  for~$\faceA$, since it has $m-1$~vertices and at least two heavy
  vertices.  The latter follows since $v_1$ has at least the arrow~$e$
  and is therefore heavy, so $v_1$ and~$v_2$ cannot be the two only
  heavy vertices of~$\face$.
    
  We set $t_x = 1$ if case~(a) happens, and $t_x = 0$ otherwise;
  similarly we set $t_y = 1$ if case~(b) happens and $t_y = 0$
  otherwise.  Or original face~$\face$ has
  \[
  \Delta := a_{2e} + a_{23} + a_{12} + a_{31} + t_x + t_y
  \]
  arrows more than the new face~$\faceA$.
    
  We collect a few more properties:
  \begin{compactenum}
  \item[(B)] If $a_{23} = k-1$ then case~(a) cannot happen: Since $x$
    exits through~$e_2$, it intersects the $k-1$ arrows starting
    in~$v_2$ and exiting through~$e_3$, and in addition the edge~$e_2$
    incident to~$v_2$.
        
  \item[(C)] If $a_{23} \geq 1$ then case~(b) cannot happen: An arrow
    starting in~$v_2$ and exiting through~$e_3$ intersects all arrows
    starting in~$v_3$ as well as the edge~$e_3$, and so $a_3 < k-1$.
    But $v_3$ needs to have~$k-1$ arrows for case~(b) to occur.
        
  \item[(D)] If $a_{31}=k-1$ then case~(a) cannot happen: Since $x$
    exits through~$e_2$, it intersects the $k-1$ arrows starting
    in~$v_3$ and exiting through~$e_1$, and in addition the edge~$e_2$
    incident to~$v_3$.
        
  \item[(E)] If $a_{12} \geq k-2$ then case~(b) cannot happen: If $y$
    intersects~$e$, then it intersects~$e$, $e_1$, and $k-2$ arrows
    from~$v_1$ to~$e_2$, a $k$-fan crossing.  If $y$ does not
    intersect~$e$, then $y$ must start in $v_2, \ldots, v_{\ell+1}$,
    and so by~(A), $y$~is short. To cross~$e_1$, $y$~must therefore
    start in~$v_3$, and so~$y$ does not cross~$g$.
  \end{compactenum}

  We now distinguish several cases to bound~$\Delta$:
  \paragraph{Case 1:} $a_{2} =a_{2e}+a_{23} \geq k-1$.  In this
  case, $a_{12} = a_{31} = 0$, and so $\Delta = a_{2} + t_x +
  t_y$. Since $a_{2e} \leq k-1$, we have by~(B) and~(C):
  \[
  \begin{array}{lccccccccl}
    & & & a_{2} & & t_x & & t_y & & \\
    \textrm{If $a_{23} = k-1$:\qquad} & \Delta & \leq &
    2(k-1) & + & 0 & +
    & 0 & = & 2k-2. \\
    \textrm{If $1 \leq a_{23} \leq k-2$:\qquad} & 
    \Delta & \leq & (2k-3) &+& 1 &+& 0  &=& 2k-2. \\
    \textrm{If $a_{23} = 0$:\qquad} & \Delta & \leq  &(k-1) &+& 1 &+& 1
    &=& k+1 \leq 2k-2. 
  \end{array}
  \]
    
  \paragraph{Case 2:} $a_{2} \leq k-2$ and $\max(a_{12},
  a_{31})=k-1$.  If $a_{12} = k-1$ then $a_{31} = 0 $, and if
  $a_{31}=k-1$ then $a_{12}=0$.  By (D) and (E), we have
  \begin{align*}
    \Delta & = a_{2} + (a_{12}+a_{31}) + (t_x + t_y) \leq (k-2) +
    (k-1) + 1 = 2k-2.
  \end{align*}
  
  \paragraph{Case 3:} $\max\{a_{2},a_{12},a_{31}\} \leq k-2$.
  By~(E), $a_{12} = k-2$ implies $t_y = 0$, and so
  \[
  \begin{array}{lccccccccccccl}
    & & & a_{2} & & a_{12} & & a_{31} & & t_x & & t_y & & \\
    \text{If $a_{12} = k-2$:} &
    \Delta & \leq & k-2 &+& k-2 &+& k-2 &+& 1 &+& 0 &=& 3k-5. \\
    \text{If $a_{12} < k-2$:} &
    \Delta & \leq & k-2 &+& k-3 &+& k-2 &+& 1 &+& 1 &=& 3k-5.
  \end{array}
  \]
  
  We now have all the ingredients to complete the inductive argument.
  Clearly $v_1$ is heavy because of the arrow~$e$. We distinguish the
  types of~$v_2$ and~$v_3$ (note that case~1 can occur only when~$v_2$
  is heavy):
  \begin{itemize}
  \item If $v_2$ is heavy and $v_3$ is not left-light, then we
    have
    \begin{align*}
      A(h, \lambda, \nu) & \leq \max(2k-2, 3k-5) + B(h-1,
      \lambda, \nu) = B(h, \lambda, \nu).
    \end{align*}
  \item If $v_2$ is heavy and $v_3$ is left-light, then $v_3$
    might become void in~$\faceA$.  In this case~$a_{31} = a_{23}
    = 0$.  The bound in case~3 improves to~$\Delta \leq 2k-3 \leq
    2k-2$:
    \begin{align*}
      A(h, \lambda, \nu) & \leq 2k-2 + B(h-1, \lambda-1, \nu+1)
      = B(h, \lambda, \nu).
    \end{align*}
  \item If $v_2$ is right-light, then $a_2 = a_{31} = 0$.  If
    $v_2$ is left-light and $v_3$ is not left-light, then $a_2 =
    a_{12} = 0$.  Either way, in case~2 the bound improves to
    $\Delta \leq k$, and in case~3 it improves to $\Delta \leq
    k-1$. We have
    \begin{align*}
      A(h, \lambda, \nu) & \leq k + B(h, \lambda-1, \nu) = B(h,
      \lambda, \nu).
    \end{align*}
  \item If $v_2$ and $v_3$ are both left-light, then $v_3$ might
    become void in~$\faceA$.  We have $a_2 = a_3 = a_{12} = 0$.
    We are thus in case~3 and have the improved bound $\Delta \leq
    2$.  This gives
    \begin{align*}
      A(h, \lambda, \nu) & \leq 2 + B(h, \lambda-2, \nu + 1) <
      B(h, \lambda, \nu).
    \end{align*}
  \item If $v_2$ is void and $v_3$ is not left-light, then $a_2
    = 0$.  In case~2 this implies~$\Delta \leq k$, in case~3 it
    implies $\Delta \leq 2k-3$. Since $k \leq 2k-3$ we have
    \begin{align*}
      A(h, \lambda, \nu) & \leq 2k - 3 + B(h, \lambda, \nu-1) =
      B(h, \lambda, \nu).
    \end{align*}
  \item Finally, if $v_2$ is void and $v_3$ is left-light, then
    $v_3$ might become void in~$\faceA$.  We have $a_2 = a_3 = 0$.
    In case~2 this implies~$\Delta \leq k$, in case~3 it implies
    $\Delta \leq k-1$. We have
    \begin{align*}
      A(h, \lambda, \nu) & \leq k + B(h, \lambda-1, \nu) = B(h,
      \lambda, \nu).
    \end{align*}
  \end{itemize}
  This completes the inductive step.
\end{proof}

\section{The general bound}
\label{sec:general}

We now prove Theorem~\ref{theo:gmain}.  The proof makes use of the
following lemma by Pach \etal \cite{pss-acn-94}:
\begin{lemma}[{{\cite[Theorem~2.1]{pss-acn-94}}}]
  \label{lemma:pach}%
  Let $\Graph$ be a graph with $n$ vertices of degree~$d_1,\dots,d_n$
  and crossing number~$\chi$.  Then there is a subset $E$ of $b$ edges
  of $\Graph$ such that removing $E$ from $\Graph$ creates components
  of size at most~$2n/3$, and
  \begin{align*}
    b^{2} \leq (1.58)^{2}\big(16\chi + \sum_{i = 1}^{n}
    d_{i}^{2}\big).
  \end{align*}
\end{lemma}

\begin{proof}[Proof of Theorem~\ref{theo:gmain}]
  Let $\Graph$ be a graph on~$n$ vertices with $m$ edges having
  property~$\edgepro$.  Since each edge~$e$ crosses a graph that has
  property~$\grapro$, the crossing number of~$\Graph$ is at most
  $\chi\leq O(mn^{1+\alpha})$.  The degree of any vertex is bounded
  by~$n-1$, and so we have $d_{i}^{2} \leq n \cdot d_i$.  It follows
  using Lemma~\ref{lemma:pach} that there exists a set~$E$ of~$b$
  edges in~$\Graph$ such that
  \begin{align*}
    b^{2} & \leq O(\chi + \sum_{i=1}^{n} d_i^{2}) \leq
    O(mn^{1+\alpha} + n \sum_{i=1}^{n}d_{i}) \leq O(mn^{1+\alpha}
    + mn) \leq O(mn^{1+\alpha}),
  \end{align*}
  and removing $E$ from $\Graph$ results in components of size at
  most~$2n/3$.
    
  We recursively subdivide $\Graph$.  Level~$0$ of the subdivision is
  $\Graph$ itself.  We obtain level~$i+1$ from level~$i$ by
  decomposing each component of level~$i$ using
  Lemma~\ref{lemma:pach}.
    
  Consider a level~$i$. It consists of $k$~components
  $G_{1},\dots,G_{k}$.  Component~$G_{j}$ has $n_{j}$ vertices and
  $m_{j}$ edges, where $n_{j}\leq (\frac 23)^{i}n$.  The total number
  of edges at level~$i$ is $r = \sum_{j=1}^{k} m_{j}$.
    
  Using the Cauchy-Schwarz inequality for the vectors $\sqrt{m_j}$,
  $\sqrt{n_{j}}$, we have
  \begin{align*}
    \sum_{j=1}^{k} \sqrt{m_jn_{j}} \leq \sqrt{\sum_{j=1}^{k}
      m_{j}} \sqrt{\sum_{j=1}^{k} n_{j}} = \sqrt{r n} \leq
    \sqrt{m n}.
  \end{align*}
  
  We first consider the case~$\alpha > 0$. The number of edges needed
  to subdivide $G_{j}$ is $O(\sqrt{m_{j}n_{j}^{1+\alpha}})$.  We bound
  this using $n_{j} \leq (\frac 23)^{i}n$ as
  $O(\sqrt{m_{j}n_{j}}((\frac 23)^{i}n)^{\alpha/2})$, and we obtain
  that the total number of edges removed between levels~$i$ and~$i+1$
  is bounded by $O(\sqrt{mn}((\frac 23)^{i}n)^{\alpha/2})$.  Since
  $(\frac 23)^{\alpha/2} < 1$, summing over all levels results in a
  geometric series, and so the total number of edges removed is
  $O(\sqrt{mn^{1+\alpha}})$.  But this implies that the total number
  of edges in the graph is bounded as
  \begin{align*}
    m \leq O(\sqrt{mn^{1+\alpha}}),
  \end{align*}
  and squaring both sides and dividing by~$m$ results in
  \begin{align*}
    m \leq O(n^{1+\alpha}).
  \end{align*}
  
  Next, consider the case $\alpha = 0$. The number of edges removed
  between levels~$i$ and~$i+1$ is bounded by $O(\sqrt{m n})$.  Adding
  over all $O(\log n)$ levels shows that
  \begin{align*}
    m & \leq O(\sqrt{mn} \log n).
  \end{align*}
  Again, squaring and dividing by~$m$ leads to
  \begin{align*}
    m & \leq O(n \log^{2} n). \qedhere
  \end{align*}
  \aftermathA
\end{proof}

\section{Conclusions}
\label{sec:conclusions}

We have proven bounds on the number of edges of $k$-fan-crossing free
graphs.  For $k = 2$ our bound is tight, and we could even
characterize the extremal graphs.  In comparison, the bound we obtain
for $k > 2$ is much weaker.  For $k = 3$, for instance, our bound
is~$6n-12$.  The best lower bound we are aware of is the construction
of Pach and \Toth~\cite{pt-gdfce-97} of a 2-planar graph with
$5n-10$~edges.  For general~$k \geq 2$, a lower bound of $n(k-1)/2$
follows by considering a $(k-1)$-regular graph, that is a graph where
every vertex has degree~$k-1$.  Clearly it is $k$-fan-crossing free
and has $n(k-1)/2$ edges.  For $k \ll n$ this can be improved by a
factor two: Consider a $\sqrt{n} \times \sqrt{n}$ integer grid of
vertices.  Connect every vertex to $2(k-1)$ neighbors within a
$O(\sqrt{k}) \times O(\sqrt{k})$ grid in a symmetric way---that is,
whenever we connect $(x,y)$ with $(x+s, y + t)$, we also connect with
$(x-s, y-t)$. Then no line can intersect more than $k-1$ edges
incident to a common vertex, and so the graph is $k$-fan-crossing
free.  Except for vertices near the boundary of the grid, every vertex
has degree~$2(k-1)$, so the number of edges is $(k-1)(n -
O(\sqrt{nk})) \approx kn$.  By contrast, our upper bound is $\approx
3kn$.

The weakness in our technique is that it analyzes arrows separately
for every face of~$H$.  When $H$ is a triangulation, it has $3n-6$
edges and $2n-4$ triangles.  Each triangle can have $3k-6$ arrows, so
we could have $(2n-4)(3k-6)$ arrows, implying $3n-6 + (2n-4)(3k-6)/2 =
3(k-1)(n-2)$ edges.  Our bound is thus the best bound that can be
obtained with this technique.  For $k = 3$, it is possible to improve
it slightly by observing that two adjacent triangles of~$H$ can only
have four arrows in total.

A $(k-1)$-planar graph is automatically $k$-fan-crossing
free. Lemma~\ref{lemma:maximal-quad} implies that an extremal
$2$-fan-crossing free graph is $1$-planar.  Is the same statement true
for $k=3$?  It certainly doesn't hold for large~$k$, as Pach and
\Toth's bound on the number of edges of $k$-planar graphs is
only~$O(\sqrt{k}n)$~\cite{pt-gdfce-97}.  Already for $k=4$,
$3$-planarity is a stronger condition than absence of a $4$-fan
crossing: Pach \etal~\cite{prtt-iclfm-06} showed that $3$-planar
graphs have at most $5.5(n-2)$ edges.  A $4$-fan-crossing free graph
with~$6n-12$ edges can be constructed by starting with a
triangulation, and adding the ``dual'' of every edge: for every pair
of adjacent triangles, connect the two vertices not shared between the
triangles.
 
The crossing number of a $1$-planar graph on $n$~vertices is at most
$n-2$ (this implies the bound $4n-8$ on the number of edges).  In
contrast, the crossing number of a fan-crossing free graph can be
quadratic.  For instance, start with the complete graph~$K_q$ on
$q$~vertices.  It has $q \choose 2$ edges and crossing
number~$\Omega(q^{4})$.  Now subdivide every edge into three edges by
inserting two vertices very close to the original vertices.  The
resulting graph is fan-crossing free and has $n = q + 2{q \choose 2}$
vertices. Since any drawing of this graph can be converted into a
drawing of~$K_q$ with the same number of crossings, it has crossing
number $\Omega(q^{4}) = \Omega(n^{2})$. (The same construction could
be done with any graph whose crossing number is quadratic in the
number of edges, such as an expander graph or even a random graph.) 

A natural next question to ask is if our techniques can be used for
graphs that do not contain a radial $(p,q)$-grid for $q > 1$, and if
we can find tighter bounds than Pach \etal~\cite{ppst-tgnlg-05}.

In Theorem~\ref{theo:gmain}, we have given a rather general bound on
the number of edges of graphs that exclude certain crossing patterns.
The theorem shows that for graph properties~$\grapro$ that imply that
the number of edges grows as $\Theta(n^{1+\alpha})$, for $\alpha > 0$,
the size of the entire graph is bounded by the same function.  For the
interesting case $\alpha = 0$, which arises for instance for
fan-crossing free graphs, our bound includes an extra $\log^{2}
n$-term.  Is this term an artifact of our proof technique, or are the
examples of graph properties where $\grapro$ implies a linear number
of edges, but graphs with~$\edgepro$ and a superlinear number of edges
exist?

\paragraph{Acknowledgments.}
For helpful discussions, we thank David Eppstein, J\'anos Pach, Antoine
Vigneron, Yoshio Okamoto, and Shin-ichi Tanigawa, as well as the other
participants of the Korean Workshop on Computational Geometry~2011 in
Otaru, Japan.

\bibliographystyle{alpha}
\bibliography{fcproof}

\end{document}